\documentclass[letter,11pt,leqno]{amsart}
\usepackage[margin=0.75 in,bindingoffset=0cm,nofoot]{geometry}
\usepackage{amsmath, amssymb, amsthm, stmaryrd, tabu, xcolor}
\usepackage{euscript}
\usepackage{verbatim}
\usepackage{graphicx}

\numberwithin{equation}{section}

\newcommand{\ghat}{\hat{g}}

\newcommand{\del}{\delta}
\newcommand{\bX}{\mathbf{A}}

\newcommand{\bS}{\mathbb{S}}

\newcommand{\inn}[2]{\left\langle  #1 \;, \; #2 \right\rangle}

\newtheorem{theorem}{Theorem}[section]
\newtheorem{prop}[theorem]{Proposition}

\newtheorem{lemma}[theorem]{Lemma}

\newtheorem{remark}[theorem]{Remark}

\newtheorem{question}[theorem]{Question}

\renewcommand{\epsilon}{\varepsilon}
\newcommand{\vep}{\epsilon}

\DeclareMathOperator{\ric}{Ric}

\DeclareMathOperator{\divergence}{div}

\newcommand{\halpha}{\widehat{\alpha}}
\newcommand{\hbeta}{\widehat{\beta}}
\newcommand{\hY}{\widehat{Y}}
\newcommand{\hv}{\widehat{v}}

\begin{document}

\title[Near-horizon geometries]{Deformations of the Kerr-(A)dS Near Horizon Geometry}
\author{Eric Bahuaud}
\address{EB: Department of Mathematics, Seattle University, Seattle, WA, United States}
\email{bahuaude@seattleu.edu}
\author{Sharmila Gunasekaran}
\address{SG: The Fields Institute for Research in Mathematical Sciences, 222 College St, Toronto ON, Canada M5T 3J1}
\address{SG: Address after 1 Jan 2024: Department of Mathematics, Radboud University, Postvak 59 6500, GL Nijmegen,
The Netherlands}
\email{gunasek1@ualberta.ca}
\author{Hari K Kunduri}
\address{HKK: Department of Mathematics and Statistics and Department of Physics and Astronomy, McMaster University, Hamilton ON, Canada L8S 4M1}
\email{kundurih@mcmaster.ca}
\author{Eric Woolgar}
\address{EW: Department of Mathematical and Statistical Sciences and
Theoretical Physics Institute, University of Alberta, Edmonton AB, Canada T6G 2G1}
\email{ewoolgar@ualberta.ca}

\maketitle
\begin{abstract}
We investigate deformations of the Kerr-(A)dS  near horizon geometry (NHG) and derive partial infinitesimal rigidity results for it.
The proof comprises two parts. First, we follow the analysis of Jezierski and Kami\'nski [Gen Rel Grav 45 (2013) 987--1004] to eliminate all but a finite number of Fourier modes of linear perturbations. In the second part, we give an argument using analyticity to prove that there are no odd Fourier modes.  
\end{abstract}

\section{Introduction}

A quasi-Einstein manifold is a solution $(M,g,X)$ of the following equation 
\begin{equation}
\label{qe}
R_{ij} + \frac{(\nabla_j X_i + \nabla_i X_j)}{2} - \frac{1}{m} X_i X_j - \lambda g_{ij} =0, \quad \lambda \in \mathbb{R} \text { and } m \in {\mathbb R}\backslash \{ 0 \},
\end{equation} 
where $M$ is a closed manifold, $g$ is a Riemannian metric on $M$, and $X$ is a one-form.\footnote 
{Some authors restrict $X$ to be gradient, while others permit $m$ and $\lambda$ to be functions. An extended form of the above definition includes cases denoted by $m=0$ and $m=\pm\infty$.}
When $m=2$, a quasi-Einstein manifold is called a \emph{vacuum near horizon geometry}, which will simply be referred to as an NHG here. An NHG arises as an induced ``zoomed-in'' limit at the horizon of an extreme black hole, i.e., a black hole with vanishing surface gravity. Four dimensional extreme black holes with a non-vanishing cosmological constant $\lambda$, i.e., de Sitter (dS) black holes if $\lambda>0$ and anti-de Sitter (AdS) black holes if $\lambda<0$, 
give rise to NHGs  with $\lambda \neq 0$ on ${\mathbb S}^2$.

Classical uniqueness theorems for stationary, axisymmetric black holes don't apply in the extreme case, i.e., black holes with degenerate horizons. Rather, in the asymptotically flat case, uniqueness of the extreme Kerr black hole was proved in \cite{FL09,AHMR09,CN10}. The proof proceeds by a global reduction to a harmonic map problem and makes use of the uniqueness of the corresponding NHG amongst all axisymmetric solutions, i.e., the Kerr NHG on $\mathbb{S}^2$. As remarked in \cite{KL13}, the $\lambda \neq 0$ case does not lend itself to the methods used to prove uniqueness in the asymptotically flat setting. Hence, when $\lambda\neq 0$ there are no known analogues for the uniqueness theorems even with non-extremality. At least in the extreme case, one can then hope to use the NHG to probe the classification/uniqueness issues by investigating possible horizon geometries. 

In fact, the NHG equations allow for other topologies. But it has been shown in \cite{DKLS18} that all NHGs for higher genus topologies are Einstein, i.e., $X \equiv 0$. 

The uniqueness of the Kerr-(A)dS NHG in the family of non-Einstein NHGs with nongradient $X$ (this is called the ``non-static'' case) was very recently proved in \cite{DL23}. As it is very relevant to our current problem, we restate their result here. 

\begin{theorem}[{\cite[Theorem 1.2]{DL23}}]
\label{theorem1.1}
Every non-static NHG $(\mathbb{S}^2,g,X)$ arises from an extreme Kerr-(A)dS black hole.
\end{theorem}

However, one cannot draw global conclusions about the underlying spacetime from knowledge of the NHG alone.
The limiting procedure that produces an NHG from a spacetime loses information about the asymptotic structure of the corresponding extreme black hole. Further, through this procedure, there is also an enhanced symmetry that is manifested in the NHGs (albeit actually absent in the parent spacetime; cf.~\cite{KL13} and the references therein).

Hence, the following (restricted) uniqueness question is still open.
 
 \begin{question}
 \label{question2}
Is the extreme Kerr-(A)dS$_4$ family the unique four-dimensional, stationary black hole solution of the Einstein vacuum equations with a cosmological constant, containing a connected degenerate Killing horizon with spherical horizon cross-sections?
\end{question}

The analogous question for static near horizon geometries for $\lambda >0$ was recently answered in \cite{KL23} establishing the uniqueness of extreme Schwarzschild de-Sitter black hole. In light of Theorem \ref{theorem1.1}, it is natural to ask if the Kerr-(A)dS NHG is an isolated solution in the linearization of equation \eqref{qe} (with $m=2$) in the moduli space of NHGs on $\mathbb{S}^2$. It is well-known from the study of Einstein manifolds that solutions of the linearized Einstein equation integrate to Einstein metrics whenever an obstruction (called formal integrability) vanishes (see \cite{Besse87}).
As an early approach to this question, the authors of \cite{JK13} searched for solutions of the linearization of \eqref{qe} (with $m=2$) about the Kerr NHG and proved that the space of such deformations was at most finite dimensional (a computer algebra analysis by \cite{CST18} eliminated the remaining deformations). Herein, we address the case of arbitrary cosmological constant. It is important to understand the linearized problem and its integrability, even in light of Theorem \ref{theorem1.1}, since the linearized problem is expected to be an important tool to study NHGs with $n\ge 3$ and with nonspherical topology. In the present work we establish the following results. 

\begin{theorem}
\label{theorem3}
There is at most a finite-dimensional space of solutions of the linearization of \eqref{qe} about the Kerr-(A)dS NHG.
For the Kerr-dS case this space is at most 11-dimensional.
\end{theorem}

Hence, any nontrivial solution $(\mathbb{S}^2,g,X)$ of the linearization of \eqref{qe} about the Kerr-(A)dS$_4$ NHG 
belongs to the moduli space of Kerr-(A)dS$_4$ NHGs or corresponds to one of the finitely many Fourier modes corresponding to infinitesimal deformations. The counting includes axisymmetric deformations. The extreme Kerr metrics with fixed cosmological constant comprise a 1-parameter family of axisymmetric deformations corresponding to the choice of rotation parameter (or equivalently the mass, since the mass, rotation parameter, and cosmological constant obey a relation; see \S \ref{sec2}). In \S \ref{sec:KerrAdS} we find that the space of axisymmetric linear deformations is one-dimensional, and hence each such a deformation is integrable.

Following the method of \cite{JK13}, we establish the theorem through a Fourier decomposition of the metric $g$ and 1-form $X$ in Fourier modes in the azimuthal angle on ${\mathbb S}^2$. We are able to eliminate all but finitely many of these modes. For the AdS case, the number of modes that we are unable to eliminate varies with the angular momentul parameter $a$; for $\lambda\ge 0$, the number of modes that we cannot eliminate does not depend on $a$. In both cases we can also eliminate odd-number modes through another argument.  It remains open whether the remaining modes can be eliminated as solutions of the linearized equations, or exist but fail to integrate to full solutions of \eqref{qe} due to obstructions.

In \S \ref{sec2}, we introduce the Kerr-(A)dS NHG and review some basic facts about this solution. We state and prove some identities for the one-form $X$ in \S \ref{sec3}. \S \ref{sec:lin} reviews the strategy for linearization. This follows the approach in \cite{JK13}, but is stated in full generality and can be adapted to any two-dimensional closed quasi-Einstein manifold. In \S \ref{sec:KerrAdS} we prove Theorem \ref{theorem3} in two parts. In the Appendix, we use some of the arguments from \S \ref{sec:KerrAdS} and apply them to the Kerr NHG eliminating the zero and odd modes. 

\subsection*{Acknowledgements} The research of EB was partially supported by a Simons Foundation Grant (\#426628, E. Bahuaud). The research of HK was supported by NSERC grant RGPIN--2018--04887. The research of EW was supported by NSERC grant RGPIN--2022--03440.  We thank the Fields Institute for their hospitality.

\section{Background: The Kerr-(A)dS near horizon geometry}
\label{sec2}

The Kerr-AdS$_4$ \cite{GLPP} near horizon geometry is a two-parameter family of solutions to equation \eqref{qe} with $m=2$ and $\lambda =-3/\ell^2$ on $M=\bS^2$. The parameter $\ell > 0$ denotes the radius of curvature at infinity and $a$ is the angular momentum parameter, where $0 < a^2 < \ell^2$. The metric takes the form
\begin{equation}
\begin{split}
g =&\, \frac{\left ( r^2 + a^2 \cos^2\theta\right )}{\left ( 1 - \frac{a^2}{\ell^2}  \cos^2\theta \right )} d\theta^2 + \frac{\sin^2\theta \left ( 1 - \frac{a^2}{\ell^2}  \cos^2\theta \right ) (r^2 + a^2)^2}{\left ( r^2 + a^2 \cos^2\theta\right ) \left ( 1 - \frac{a^2}{\ell^2}\right )^2} d\phi^2, \\ 
X =&\, \frac{2a^2 \cos\theta \sin \theta}{\left ( r^2 + a^2 \cos^2\theta\right )} d\theta + \frac{2 a r \sin^2\theta \left ( 1 - \frac{a^2}{\ell^2}  \cos^2\theta \right ) \left (r^2 + a^2\right )}{\left ( r^2 + a^2 \cos^2\theta\right )^2 \left ( 1 - \frac{a^2}{\ell^2}\right )} d\phi .
\label{KerrAdSmetric}
\end{split}
\end{equation} 
Here $\theta \in (0,\pi)$ and $\phi$ is required to be $2\pi$-periodic. Observe that
\begin{equation}
|X|^2 =\frac{ 4a^2 \left ( 1 - \frac{a^2}{\ell^2}  \cos^2\theta \right ) \sin^2\theta}{\left ( r^2 + a^2 \cos^2\theta\right )^2}.
\end{equation} and in particular, $X$ vanishes at the poles $\theta = 0$ and $\theta = \pi$. The restriction $a^2 < \ell^2$ ensures that $\frac{a^2}{\ell^2}<1$, and $\frac{a^2}{\ell^2} \cos^2\theta  <1$. The parameters  $(r,a)$ are both positive, and are related by 
\begin{equation} \label{eqn:rel-r-a}
\begin{split}
0=&\, r^2 - a^2 + \frac{r^2}{\ell^2}(3r^2 + a^2),
\end{split}
\end{equation} The relation \eqref{eqn:rel-r-a} fixes one of the parameters $r$ or $a^2$ in terms of the other. Indeed, 
\begin{equation}\label{KerrAdScurve}
a^2 = {\frac { r^2 ( 3r^2 + {\ell}^{2} )}{ \left( {\ell}^{2} - {r^2} \right)}}.
\end{equation} 

Under the rescaling $r \mapsto \ell r $ and $a \mapsto \ell a$ we may remove the $\ell$ dependence in the solution. This is equivalent to fixing $\ell = 1$, so that $\lambda = -3 $ and $0 < a^2 < 1$ for simplicity. Note that 
\begin{equation}
    0 < \frac{r^2}{a^2} < 1, \;\; \text{ and } \;\; 0 < r < 1/\sqrt{3}.
\end{equation} 
We may now regard the Kerr-AdS$_4$ geometry as a one-parameter family of metrics where the parameters $(a,r)$ lie on the curve defined by \eqref{KerrAdScurve}. With the transformation $x = \cos \theta$, the near horizon geometry metric can be expressed in the form
\begin{equation}
\begin{split}
\label{metric:KerrAdS}
g =&\, \frac{\Psi(x)}{f(x)} dx^2 + \nu^2 \frac{f(x)}{\Psi(x)} d \phi^2 ,\\ 
X =&\, - \frac{2 a^2 x}{\Psi(x)} dx + \frac{2 a \nu r f(x)  d \phi}{\Psi(x)^2} ,
\end{split}
\end{equation}
where $x=\pm 1$ are poles (where $\partial_\phi$ degenerates) and
\begin{equation} \label{metric:comps}
\nu =  \frac{r^2 + a^2}{1 - a^2}, \quad \Psi(x) = r^2 + a^2 x^2, \quad  f(x) = \left( 1 - a^2 x^2 \right) (1-x^2).
\end{equation}
So we have
\begin{align}
\frac{f(x)}{\Psi(x)} = \frac{\left( 1 - a^2 x^2 \right) (1-x^2)}{r^2 + a^2 x^2}, \; \mbox{and} \; |X|^2 = \frac{4 a^2 f}{\Psi^2}.
\end{align}
Note that $f/\Psi$ has simple zeroes at the poles $x=\pm 1$. The other zeroes are outside of the domain of interest. We can retrieve the Kerr-dS NHG from the above using $\ell^2 \mapsto - \ell^2$. This redefines the functions appearing in the Kerr-AdS NHG as follows.
\begin{equation} \label{metric:comps:dS}
\nu =  \frac{r^2 + a^2}{1 + a^2}, \quad \Psi(x) = r^2 + a^2 x^2, \quad  f(x) = \left( 1 + a^2 x^2 \right) (1-x^2).
\end{equation}
Here $\lambda = 3$. The equivalent of \eqref{eqn:rel-r-a} in this case is
\begin{equation} \label{KerrdScurve}
    a^2 = \frac{r^2(\ell^2 - 3r^2)}{r^2 + \ell^2}.
\end{equation}
With the rescaling $r \mapsto \ell r$ and $a\mapsto \ell a$, we get the following for the ranges of $r$ and $a$
\begin{align}
    a^2 = r^2 - \frac{4r^4}{1+ r^2} < r^2 \,, \quad 0 < r < 1/\sqrt{3}\,, \quad \sup a = \sqrt{7 - 4 \sqrt{3}} \quad \text{ and } \quad r|_{\sup 
 a} = \sqrt{\frac{2}{\sqrt{3}} - 1}.
\end{align}

\section{Remarks concerning the quasi-Einstein one-form $X$}
\label{sec3}

In this section, we take up a few results concerning the covector field $X$ participating in a solution to equation \eqref{qe} on a compact surface.  Fundamental to the Jezierski-Kami\'nski method is the auxiliary one-form  given by
\[ \Phi := \frac{X}{|X|^2}, \]
where $X$ is non-vanishing.

\begin{lemma}\label{lemma2.1}
Let $(M^2,g,X)$ be a solution to equation \eqref{qe} for some $m \in (0,\infty)$ on a closed surface. Let $U\subset M$ be an open connected domain on which the 1-form $X$ has no zeroes. Then $\Phi$ is a closed one-form on $U$.
\end{lemma}
\begin{proof}
Recalling that on a surface $R_{ij} = K g_{ij}$, where $K$ is the Gauss curvature, equation \eqref{qe}, implies
\begin{equation} \label{eq2.1}
\nabla_i X_j = \frac{1}{m}X_iX_j + \left ( \lambda - K\right ) g_{ij} + \frac{1}{2} \left ( \nabla_i X_j -\nabla_j X_i\right ).
\end{equation} 
Now
\begin{equation}
\nabla_i \Phi_j -\nabla_j \Phi_i = \frac{1}{|X|^2} \left ( \nabla_i X_j -\nabla_j X_i \right ) - \frac{2}{\left ( | X|^2\right )^2} X_j X^k \nabla_i X_k + \frac{2}{\left ( | X|^2\right )^2} X_i X^k \nabla_j X_k. 
\end{equation} 
Using \eqref{eq2.1}, this becomes
\begin{align}
\nabla_i \Phi_j -\nabla_j \Phi_i &= \frac{1}{|X|^2} \left ( \nabla_i X_j -\nabla_j X_i \right )  + \frac{X^kX_i}{ \left ( |X|^2\right )^2}\left ( \nabla_j X_k-\nabla_kX_j\right) \nonumber \\
& - \frac{X^kX_j}{ \left ( |X|^2\right )^2}\left ( \nabla_i X_k-\nabla_kX_i\right ) .
\end{align}
The last two terms on the right-hand side are pairwise antisymmetric in their free indices. Any 2-form in two dimensions is proportional to the volume form; i.e., if $\beta$ is a 2-form then
\begin{equation}
\beta_{ij} = \frac{\beta_{mn}\epsilon^{mn}}{2} \epsilon_{ij},
\end{equation} 
where $\epsilon_{ij}$ represents the volume form (not the tensor density). Recall that this satisfies the identity $\epsilon^{kj} \epsilon_{ki} = \delta^{j}_{~i}$. If we write $\nabla_i X_j-\nabla_j X_i =:  \epsilon_{ij}F$ for $F := \epsilon^{ij} \nabla_iX_j$, then a calculation shows that
\begin{equation}
\left [\frac{X^kX_i}{ \left ( |X|^2\right )^2}\left ( \nabla_j X_k-\nabla_kX_j\right ) - \frac{X^kX_j}{ \left ( |X|^2\right )^2}\left ( \nabla_i X_k-\nabla_kX_i\right )\right ] = -\frac{2F}{|X|^2} \epsilon_{ij},
\end{equation} 
and then $\epsilon^{ij} \left ( \nabla_i \Phi_j -\nabla_j \Phi_i\right )=0$. Thus $d\Phi =0$ as claimed. 
\end{proof}

It will be important to know when $X$ vanishes.  While topological considerations force the vanishing of $X$ if the background manifold is a sphere, the next lemma guarantees a zero more generally.

\begin{lemma}\label{lemma2.2}
Let $(M^2,g,X)$ be a solution of equation \eqref{qe} for some $m \in (0,\infty)$ on a closed surface. Then $X$ must have a zero and
\begin{align}
\divergence \Phi=-\frac{1}{m}.
\end{align}
\end{lemma}
\begin{proof}
Working momentarily in $n$-dimensions, from the quasi-Einstein equation \eqref{qe} we have
\begin{equation}
\begin{aligned}
\divergence \Phi =&\, \frac{\divergence X}{|X|^2} -\frac{\nabla_X \left ( | X|^2\right )}{\left ( | X|^2\right )^2}\\
= &\, \frac{1}{m}+\frac{n\lambda}{| X|^2} -\frac{R}{| X|^2}-\frac{X^kX^l\left ( \nabla_k X_l+\nabla_l X_k\right )}{\left ( | X|^2\right )^2}\\
=&\, \frac{1}{m}+\frac{n\lambda}{| X|^2} -\frac{R}{| X|^2}-\left [ \frac{2}{m}+\frac{2\lambda}{| X|^2} -\frac{2\ric(X,X)}{\left ( | X|^2\right )^2}\right ]\\
=&\, -\frac{1}{m}+\frac{(n-2)\lambda}{|X|^2} +\frac{2\ric(X,X)}{\left ( | X|^2\right )^2}-\frac{R}{| X|^2}.
\end{aligned}
\end{equation}
On a surface, we have $\displaystyle \ric=\frac{R}{2}g$ and thus
\begin{equation}
\label{eq3.38}
\divergence \Phi=-\frac{1}{m}.
\end{equation}
For finite $m$ this has no solutions on any compact manifold. Hence $\Phi=X/|X|^2$ cannot be globally defined, and so $X$ must have a zero.
\end{proof}

\section{Linearization of the quasi-Einstein equation} \label{sec:lin}

In this section we begin by reviewing our strategy to prove uniqueness of the Kerr-AdS near horizon geometry.  We then compute the linearization of the quasi-Einstein equation on a surface, and set up the method of \cite{JK13}.

\subsection{A review of the strategy}

The quasi-Einstein operator can be thought of as a nonlinear smooth map
\begin{align}
\begin{aligned}
\mathcal{Q}: C^{\infty}(M; \Sigma^2( T^* M) ) \times C^{\infty}(M; T^*M) \longrightarrow C^{\infty}(M; \Sigma^2( T^* M) ), \\
\mathcal{Q}(g,X) = R_{ij} + \frac{(\nabla_j X_i + \nabla_i X_j)}{2} - \frac{1}{m} X_i X_j - \lambda g_{ij},
\end{aligned}
\end{align}
where $\Sigma^2(T^*M)$ is the bundle of symmetric $2$-tensors.  The quasi-Einstein equation \eqref{qe} now reads $\mathcal{Q}(g,X) = 0$.

Suppose that for $s \in (-\vep,\vep)$, $(g(s), X(s))$ is a smooth curve of solutions $\mathcal{Q}(g(s),X(s)) = 0$ to the quasi-Einstein equation. Differentiating in $s$, and writing $\dot{g} = g'(0)$, $\dot{X} = X'(0)$, we obtain that the linearization $\mathcal{L}$ of $\mathcal{Q}$ at $(g(0),X(0))$ satisfies
\begin{equation} 
\mathcal{L}_{(g(0),X(0))}  (\dot{g},\dot{X}) = \left. \frac{d}{ds} \mathcal{Q}(g(s),X(s)) \right|_{s=0} = 0, 
\end{equation}
i.e., $(\dot{g},\dot{X})$ is an infinitesimal deformation of the quasi-Einstein equations at $(g(0),X(0))$. Thus a smooth curve of solutions to the quasi-Einstein equations yields a nontrivial kernel to this linear operator. The strategy to prove uniqueness statements is thus to show that a putative element of this kernel must be trivial, or in the case of Theorem \ref{theorem3}, that the space of solutions (modulo diffeomorphisms) must be finite dimensional.

Since we work on ${\mathbb S}^2$, a conformal approach to restrict the diffeomorphism freedom is natural and will simplify the linearization. It is possible to express the metric variation $g(s)$ as a conformal multiple of the `background' metric $g(0)$. In fact, slightly more is true. Thanks to the geometric nature of equation \eqref{qe}, we may express the metric variation $g(s)$ as a conformal multiple of the pullback of $g(0)$ by an $s$-dependent diffeomorphism. This additional diffeomorphism will be used in \S \ref{sec:KerrAdS} to set initial conditions by allowing us to fix the locations where the norm of the covector field $X(s)$ vanishes.

We will study the general linearization of the quasi-Einstein equations on surfaces in \S \ref{sec:linearization-surfaces} and then discuss the method of \cite{JK13} in \S \ref{sec:JK}. We specialize to the case of the Kerr-(A)dS near horizon geometries in \S \ref{sec:KerrAdS}.


\subsection{Linearization computation on surfaces} \label{sec:linearization-surfaces}

Assume that $(M^2, g(s), X(s))$ satisfies $\mathcal{Q}(g(s),X(s)) = 0$ for $s \in (-\vep,\vep)$. 

In order to compute the linearization we first write $g(s)$ as a conformal variation in terms of $g_0 = g(0)$.  Since $M$ is a homeomorphic to a sphere, for each $s$ we may solve for a smooth conformal factor $\Omega(s)$ so that $g(s) = e^{2 \Omega(s)} g_{\star}$, with $g_{\star}$ a metric with constant positive curvature on the sphere.  Since this is also true for $s=0$, we can also write $g(s) = e^{2 (\Omega(s) - \Omega(0))} g_0$.

In order to help us choose boundary conditions for the partial differential equations satisfied by the linearization\footnote{This is done for the Kerr-(A)dS family of metrics at the beginning of \S \ref{sec:KerrAdS}.}, it will also be helpful to allow for an $s$-dependent diffeomorphism of $M$ that allows us to fix the locations of the vanishing of $|X(s)|$.  Thus set $u(s) = \Omega(s) - \Omega(0)$ and assume $\theta_s: M \to M$ is a smooth family of smooth diffeomorphisms so that we may express the metric as
\begin{equation} \label{eqn:conf}
g(s) = e^{2 u(s)} \theta_s^* g_0.
\end{equation}

Let us re-express equation \eqref{qe} in terms of the metric $g_0$ in equation \eqref{eqn:conf}. Introduce $\ghat_s = \theta_s^* g_0$ as a temporary device.  The conformal transformation of the metric has no effect on the one-form $X$. The covariant derivative of a one-form transforms as
\begin{equation}
\nabla^{g(s)}_i X_j = \nabla^{\ghat_s}_i X_j + \inn{X}{du}_{\ghat_s} (\ghat_s)_{ij} - X_j \nabla^{\ghat_s}_i u - X_i \nabla^{\ghat_s}_j u. 
\end{equation}
The equation for Gaussian curvature under a conformal change of metric allows us to rewrite the Ricci curvature term of $g(s)$ as
\begin{equation} \label{eqn:curv-inv}
R_{ij}(g(s)) = K(g(s)) g_{ij}(s) = \left(K(\ghat_s) - \Delta_{\ghat_s} u(s) \right) (\ghat_s)_{ij}.
\end{equation}

Now, naturality of curvature implies that equation \eqref{eqn:curv-inv} may be written
\begin{equation} 
R_{ij}(g(s)) = \theta_s^* \left(K(g_0) - \Delta_{g_0} \widehat{u}(s) \right) (g_0)_{ij},
\end{equation}
where we have introduced $\widehat{u}(s) = u(s) \circ \theta_s^{-1}$.  In a similar way, by defining the one-form $\widehat{X} = (\theta_s^{-1})^* X$, we may express the remaining terms in equation \eqref{qe} as the pullback via $\theta_s$.  
\begin{align}
\begin{aligned}
0=&\theta_s^* \left( \left(K(g_0) - \Delta_{g_0} \widehat{u}(s) \right) (g_0)_{ij} - \lambda e^{2 \widehat{u}(s)} (g_0)_{ij} 
+\frac{1}{2} \nabla^{g_0}_j \widehat{X}_i(s) + \frac{1}{2} \nabla^{g_0}_i \widehat{X}_j(s) \right. \\
& \left. + \inn{ \widehat{X}(s)}{d\widehat{u}(s)}_{g_0} (g_0)_{ij} - \widehat{X}(s)_i \nabla^{g_0}_j \widehat{u}(s) - \widehat{X}_j(s) \nabla^{g_0}_i \widehat{u}(s) - \frac{1}{m} \widehat{X}_i(s) \widehat{X}_j(s) \right). 
\end{aligned}
\end{align}

We may thus compose with the inverse diffeomorphism to obtain the equation $\mathcal{Q}(g(s),X(s))=0$ in terms of $g_0$.  In order to lighten the notation, we now drop the hat from $u$ and $X$, recalling that these quantities are composed with the diffeomorphism. We obtain
\begin{align} \label{eqn:conformal-QE}
\begin{aligned}
0=&\left(K(g_0) - \Delta_{g_0} u(s) \right) (g_0)_{ij} - \lambda e^{2 u(s)} (g_0)_{ij} 
+\frac{1}{2} \nabla^{g_0}_j X_i(s) + \frac{1}{2} \nabla^{g_0}_i X_j(s) \\
& + \inn{ X(s)}{du(s)}_{g_0} (g_0)_{ij} - X(s)_i \nabla^{g_0}_j u(s) - X_j(s) \nabla^{g_0}_i u(s) - \frac{1}{m} X_i(s) X_j(s). 
\end{aligned}
\end{align}
We emphasize that all curvatures, covariant derivatives, and metric contractions are with respect to $g_0$ in this equation.

We now linearize by differentiating in $s$ and evaluating at $s=0$. The only $s$-dependence in the equation above occurs through $u$ and $X$, thus we introduce the variation variables
\begin{align}
v := u'(0), \; \; \; Y := X'(0)
\end{align}
for the first-order perturbations of these quantities.  The linearization of equation \eqref{eqn:conformal-QE} is thus
\begin{align} \label{eqn:linearization1}
\begin{aligned}
0&= - \Delta_{g_0} v \, (g_0)_{ij} - 2 \lambda v \, (g_0)_{ij}  + \frac{1}{2} \nabla^{g_0}_j Y_i + \frac{1}{2} \nabla^{g_0}_i Y_j \\&+  \inn{ X}{dv}_{g_0} (g_0)_{ij} - X_i \nabla^{g_0}_j v - X_j \nabla^{g_0}_i v - \frac{1}{m} X_i Y_j - \frac{1}{m} Y_i X_j. 
\end{aligned}
\end{align}

Having completed the linearization, we set $g = g_0$ to lighten the notation, with the understanding that all covariant derivatives and metric contractions are with respect to $g_0$.  The trace part of equation  \eqref{eqn:linearization1} can be written as

\begin{align} \label{eqn:linearization2}
0=\nabla^i ( Y_i - 2 \nabla_i v ) - 4 \lambda v - \frac{2}{m} \inn{X}{Y}.
\end{align}
The tracefree part of the equation \eqref{eqn:linearization1} may be written as
\begin{equation} \label{eqn:linearization3}
\begin{split}
0=&\, \frac{1}{2} \nabla_j Y_i + \frac{1}{2} \nabla_i Y_j - \frac{1}{2} \nabla^k Y_k g_{ij} +  \inn{ X}{dv} g_{ij} -  X_i \nabla_j v -  X_j \nabla_i v - \frac{1}{m} X_i Y_j\\
&\, - \frac{1}{m} Y_i X_j  + \frac{1}{m} \inn{X}{Y} g_{ij}.
\end{split}
\end{equation}

\subsection{The method of \cite{JK13}} \label{sec:JK} 

So far the computations of the previous section have been completely general to a quasi-Einstein metric on a compact surface.  We now follow the method of   Jezierski-Kami\'nski and cast it slightly more generally.  For this part of the argument, we need only assume the background quasi-Einstein metric $(g,X)$ can be expressed in local coordinates in diagonal form,
\begin{align} \label{eqn:metric-ansatz}
g = g_{xx}(x,\phi) dx^2 + g_{\phi \phi}(x,\phi) d\phi^2,
\end{align}
with nonvanishing components on its domain.  We specialize to the Kerr-(A)dS metrics in a later section.

We write $\Phi = X / |X|^2$, and assume that $\Phi$ is defined on the domain of the local coordinates above. The method of \cite{JK13} is to project the variation vector field $Y$ onto $\Phi$ and its orthogonal complement. To this end, we introduce functions
\begin{equation}
\label{alphabeta}
\begin{split}
\alpha :=& \vep^{ij} \Phi_i Y_j = \nu^{-1} (\Phi_x Y_{\phi} - \Phi_{\phi} Y_x), \\
\beta :=& \Phi^j Y_j = \Phi^x Y_x + \Phi^{\phi} Y_{\phi}.
\end{split}
\end{equation}
The next proposition summarizes the rewriting of the the linearization equations in terms of these quantities. 

\begin{prop}
Given the curve of quasi-Einstein metrics of Section \ref{sec:linearization-surfaces} expressed as in equation \eqref{eqn:conf}, the linearized equations \eqref{eqn:linearization2} and \eqref{eqn:linearization3} imply
\begin{equation} \label{eqn:perturbation}
\begin{gathered}
0=\Delta \alpha +  \frac{1}{m} \vep^{ij}\nabla_j Y_i,\\
-4 \lambda v = \Delta \beta - \left( 1 + \frac{1}{m} \right) \nabla^i Y_i + \frac{2}{m} |X|^2 \beta ,\\
\nabla_i v = -\frac{1}{2} {\vep_i}^j \nabla_j \alpha + \frac{1}{2} \nabla_i \beta - \frac{1}{2m} Y_i,
\end{gathered}
\end{equation}
where $\alpha$ and $\beta$ are defined in \eqref{alphabeta}.  Note all derivatives and metric contractions are with respect to $g_0$.
\end{prop}
\begin{proof}
The proof is similar to \cite[Appendix B]{JK13}, but we reproduce it here for the convenience of the reader. The metric ansatz of equation \eqref{eqn:metric-ansatz} is used extensively in the calculations that follow to raise and lower indices.  

The first step is to take the tracefree part of the linearized equation \eqref{eqn:linearization3} and express the derivative of $v$ in terms of $Y$.  To begin, first take the $xx$ component of equation \eqref{eqn:linearization3}
\begin{align}
\nabla_x Y_x - \frac{1}{2} \nabla^k Y_k g_{xx} +  \inn{ X}{dv} g_{xx} - 2 X_x \nabla_x v - \frac{2}{m} X_x Y_x + \frac{1}{m} \inn{X}{Y} g_{xx} = 0,
\end{align}
and dividing by $g_{xx}$ and using the explicit diagonal structure to raise indices as necessary, we write this as
\begin{align}
\frac{1}{2} \nabla^x Y_x - \frac{1}{2} \nabla^{\phi} Y_{\phi}  - X^x \nabla_x v + X^{\phi} \nabla_{\phi} v - \frac{1}{m} X^x Y_x + \frac{1}{m} X^{\phi} Y_{\phi} = 0,
\end{align}
and finally we express this equation as 
\begin{align} \label{eqn:linearization4}
 X^x \nabla_x v -  X^{\phi} \nabla_{\phi} v =  \frac{1}{2}\nabla^x Y_x - \frac{1}{2} \nabla^{\phi} Y_{\phi}  -\frac{1}{m} X^x Y_x + \frac{1}{m} X^{\phi} Y_{\phi}.
\end{align}

Next, we take the $x\phi$ component of equation \eqref{eqn:linearization3}, and use that the metric is diagonal to obtain
\begin{align}
\frac{1}{2}\nabla_{\phi} Y_x + \frac{1}{2}\nabla_x Y_{\phi}  -  X_x \nabla_{\phi} v -  X_{\phi} \nabla_x v - \frac{1}{m} X_x Y_{\phi} - \frac{1}{m} Y_x X_{\phi}  = 0,
\end{align}
which we rewrite as
\begin{align} \label{eqn:linearization5}
 X_{\phi} \nabla_x v  + X_x \nabla_{\phi} v = \frac{1}{2} \nabla_{\phi} Y_x + \frac{1}{2} \nabla_x Y_{\phi} - \frac{1}{m} X_x Y_{\phi} - \frac{1}{m} Y_x X_{\phi}.
\end{align}
Equations \eqref{eqn:linearization4} and \eqref{eqn:linearization5}, can then be solved for $\nabla_x v$ and $\nabla_{\phi} v$. Indeed we may write
\begin{gather}
\begin{bmatrix}  X_\phi &  X_x \\ $ $ \\  X^x & -  X^\phi  \end{bmatrix} \begin{bmatrix} \nabla_x v \\ $ $ \\ \nabla_{\phi} v  \end{bmatrix}
= 
\begin{bmatrix}
\displaystyle \frac{1}{2} \nabla_\phi Y_x  +\frac{1}{2}  \nabla_x Y_\phi
- \frac{1}{m} X_x Y_\phi - \frac{1}{m} X_\phi Y_x \\ $ $ \\
\displaystyle \frac{1}{2}\nabla^x Y_x  - \frac{1}{2} \nabla^\phi Y_\phi  - \frac{1}{m} X^x Y_x + \frac{1}{m} X^\phi Y_\phi
\end{bmatrix}.
\end{gather}

Now, if $A = \begin{bmatrix}  X_\phi &  X_x \\ $ $ \\  X^x & -  X^\phi  \end{bmatrix} $, then $A^{-1} = \frac{1}{|X|^2  } \begin{bmatrix}   X^\phi & X_x \\ $ $ \\ X^x & - X_\phi  \end{bmatrix}$. 
Thus, the derivatives of $v$ are
\begin{gather}
\begin{bmatrix} \nabla_x v \\ $ $ \\ \nabla_{\phi} v  \end{bmatrix}
= \frac{1}{|X|^2  } \begin{bmatrix}   X^\phi & X_x \\ $ $ \\ X^x & - X_\phi\end{bmatrix}
\begin{bmatrix}
\displaystyle \frac{1}{2} \nabla_\phi Y_x  +\frac{1}{2}  \nabla_x Y_\phi
- \frac{1}{m} X_x Y_\phi - \frac{1}{m} X_\phi Y_x \\ $ $ \\
\displaystyle \frac{1}{2}\nabla^x Y_x  - \frac{1}{2} \nabla^\phi Y_\phi  - \frac{1}{m} X^x Y_x + \frac{1}{m} X^\phi Y_\phi
\end{bmatrix}.
\end{gather}
It is then straightforward to check that this matrix equation is equivalent to 
\begin{align} \label{eqn:solved-for-v1}
\nabla_i v  = \frac{1}{|X|^2}\left( \frac{1}{2} X^j \nabla_j Y_i + \frac{1}{2} X^j {\nabla}_i Y_j - \frac{1}{2} X_i {\nabla}^j Y_j - \frac{1}{m} | X|^2 Y_i \right), 
\end{align}
which shows that the derivatives of the variation of the conformal factor can be expressed entirely in terms of $Y$ and the background data.

The next step is to express the equation for the derivatives of $v$ in terms of the  potential function $\Phi = X/|X|^2$.  Begin by writing equation \eqref{eqn:solved-for-v1} using $\Phi$,
\begin{align} \label{eqn:solved-for-v2}
\nabla_i v  = \frac{1}{2} \Phi^j \nabla_j Y_i + \frac{1}{2} \Phi^j {\nabla}_i Y_j - \frac{1}{2} \Phi_i {\nabla}^j Y_j - \frac{1}{m} Y_i, 
\end{align}
Now, observe that
\begin{align}
\begin{aligned}
\nabla_i ( \Phi^j Y_j ) &= \Phi^j \nabla_i Y_j + Y^j \nabla_i \Phi_j,\\
\nabla_j ( \Phi^j Y_i ) &= \Phi^j \nabla_j Y_i + Y_i \nabla^j \Phi_j,\\
\nabla^j ( \Phi_i Y_j ) &= \Phi_i \nabla^j Y_j + Y_j \nabla^j \Phi_i.
\end{aligned}
\end{align}
Inserting these expressions into equation \eqref{eqn:solved-for-v2}, and using successively that $\nabla^j \Phi_j = -1/m$ by Lemma \ref{lemma2.2}, $d \Phi = 0$ and the definitions of $\alpha$ and $\beta$ we get
\begin{align} \label{eqn:solved-for-v3}
\begin{aligned}
\nabla_i v &= \frac{1}{2} \nabla_j ( \Phi^j Y_i ) - \frac{1}{2} Y_i \nabla^j \Phi_j + \frac{1}{2} \nabla_i ( \Phi^j Y_j ) - \frac{1}{2} Y^j \nabla_i \Phi_j - \frac{1}{2} \nabla^j ( \Phi_i Y_j ) + \frac{1}{2} Y_j \nabla^j \Phi_i  - \frac{1}{m} Y_i, \\
&= \frac{1}{2} \nabla^j ( \Phi_j Y_i - \Phi_i Y_j ) + \frac{1}{2} \nabla_i ( \Phi^j Y_j ) + \frac{1}{2} Y^j \left( \nabla_j \Phi_i - \nabla_i \Phi_j\right)  + \left(\frac{1}{2m} - \frac{1}{m}\right) Y_i, \\
&= \frac{1}{2} \nabla^j ( \Phi_j Y_i - \Phi_i Y_j ) + \frac{1}{2} \nabla_i ( \Phi^j Y_j ) - \frac{1}{2m} Y_i  \\
&= -\frac{1}{2} {\vep_i}^j \nabla_j \alpha + \frac{1}{2} \nabla_i \beta - \frac{1}{2m} Y_i.
\end{aligned}
\end{align}

The final step is to obtain second order equations for $\alpha$ and $\beta$.  Returning to the trace of the quasi-Einstein linearization, equation \eqref{eqn:linearization2}, we insert the expression for the derivatives of $v$ above and simplify.  A straightforward computation shows that
\begin{align} 
\Delta \beta - \left( 1 + \frac{1}{m} \right) \nabla^i Y_i + \frac{2}{m} |X|^2 \beta = -4 \lambda v.
\end{align}

A further equation is available.  Since $v$ is a smooth function, its Hessian is symmetric, and thus ${\vep^{ij}} \nabla_j \nabla_i v = 0$.  Applying the expression for $\nabla_i v$ from equation \eqref{eqn:solved-for-v3}, we obtain
\begin{align}
\begin{aligned}
0 &= \vep^{ij} \nabla_j \left(-\frac{1}{2} {\vep_i}^k \nabla_k \alpha + \frac{1}{2} \nabla_i \beta - \frac{1}{2m} Y_i \right) \\
&= -\frac{1}{2} \vep^{ij} {\vep_i}^k  \nabla_j \nabla_k \alpha + \frac{1}{2} \vep^{ij} \nabla_j \nabla_i \beta -\frac{1}{2m}  \vep^{ij} \nabla_j Y_i \\
&= -\frac{1}{2} \Delta \alpha -\frac{1}{2m} \vep^{ij}\nabla_j Y_i, 
\end{aligned}
\end{align}
which we write as
\begin{align}
\Delta \alpha +  \frac{1}{m} \vep^{ij}\nabla_j Y_i = 0.
\end{align}

In conclusion, the linearization equations can be expressed as equation \eqref{eqn:perturbation} above.
\end{proof}

\section{Variations of the Kerr-(A)dS near horizon geometry} \label{sec:KerrAdS}

In this section we specialize to Kerr-(A)dS near horizon metrics, and prove Theorem \ref{theorem3}.  To give an overview of the proof and this section, we let $(g,X)$ be such Kerr-(A)dS metric and consider a smooth one-parameter deformation $(e^{2u(s)} g(0),X(s))$ of the quasi-Einstein equation as in \S \ref{sec:lin}.  In \S \ref{sec:kerradsback} we write the perturbation equation \eqref{eqn:perturbation} with respect to this background and then compute Fourier series in the $\phi$-variable.  In \S \ref{sec:noodd} we argue that there are no nontrivial odd-order modes of the linearization, and in \S \ref{sec:zeromodes} we argue there is a one-dimensional family of perturbations which we already know corresponds to the axisymmetric solutions of the Kerr-(A)dS family.  Finally in \S \ref{sec:global} we present an argument that allows to conclude that only finitely many of the modes may be nonzero, hence proving Theorem \ref{theorem3}.  We conclude with an estimate on the number of modes for each of the Kerr-dS and Kerr-AdS cases.

We begin by treating both the dS and AdS cases simultaneously.  We again label the first variation of these quantities at $s=0$ by $(v, Y)$. Set $m = 2$, $|\ell| = 1$ so that $\lambda = \pm 3$, and 
equation \eqref{metric:KerrAdS}, with either \eqref{metric:comps} or \eqref{metric:comps:dS} to give a chart for the Kerr-(A)dS geometry on $(-1,1)_x \times \bS^1_{\phi}$. We refer to the zeroes of $X$ at $x=\pm 1$ as (north/south) poles.

As $X(0)$ has two zeroes of index one, $X(s)$ will also have two zeroes for sufficiently small $s$. Suppose that $X$ vanishes at exactly two points, $p$ and $q$.  There is a diffeomorphism $\theta: \mathbb{S}^2 \to \mathbb{S}^2$ of $g$ that maps $p$ to the north pole and $q$ to the south pole, and this is the choice we make in equation \eqref{eqn:conf}.  Under this choice, the composition of $X$ with the diffeomorphism (still denoted by $X$ hereafter) now forces $X$ to vanish at the poles.  This justifies the boundary conditions $\alpha (\pm 1) = \beta (\pm 1) = 0$   for the differential equations describing perturbations.

\subsection{The equations in Kerr-(A)dS background} \label{sec:kerradsback}

With the Kerr-(A)dS background, the Laplacian on an arbitrary (sufficiently differentiable) function $F$ takes the form
\begin{align}
\Delta F
= \partial_x \left( \frac{f}{\Psi} \partial_x F \right) + \frac{\Psi}{\nu^2 f} \partial^2_{\phi} F.
\end{align}

The perturbation equations \eqref{eqn:perturbation} under the Kerr-(A)dS background read
\begin{align}
\label{eq5.2}
\begin{aligned}
0=&\, \partial_x \left( \frac{f}{\Psi} \partial_x \alpha \right) + \frac{\Psi}{\nu^2 f} \partial^2_{\phi} \alpha +  \frac{1}{2 \nu} (\partial_{\phi} Y_{x} - \partial_{x} Y_{\phi}),\\
\mp 12 v=&\, \partial_x \left( \frac{f}{\Psi} \partial_x \beta \right) + \frac{\Psi}{\nu^2 f} \partial^2_{\phi} \beta - \frac{3}{2} \left( \partial_x \left(\frac{f}{\Psi} Y_x \right) + \frac{\Psi}{\nu^2 f} \partial_{\phi} Y_{\phi} \right) + \frac{4 a^2 f}{\Psi^2} \beta ,\\
v_x =&\, -\frac{1}{2} \frac{\Psi}{f \nu} \partial_{\phi} \alpha+ \frac{1}{2} \partial_x \beta - \frac{1}{4} Y_x,\\
v_{\phi} =&\, \frac{1}{2} \frac{\nu f}{\Psi} \partial_x \alpha + \frac{1}{2} \partial_{\phi} \beta - \frac{1}{4} Y_{\phi},
\end{aligned}
\end{align}
where the $\mp$ sign in the second equation in \eqref{eq5.2} is should be taken to be $-$ for the dS case and $+$ for the AdS case, with metric components as in either equation \eqref{metric:comps} or \eqref{metric:comps:dS}. Additionally the definitions of $\alpha$ and $\beta$ given in equation \eqref{alphabeta} in the Kerr-(A)dS background are
\begin{align} \label{alphabeta2}
\begin{aligned}
\alpha &= - \frac{x \Psi }{2 \nu f} Y_{\phi} - \frac{r}{2a} Y_x, \\
\beta &= - \frac{x}{2} Y_x + \frac{r \Psi}{2 a \nu f} Y_{\phi}.
\end{aligned}
\end{align}
We solve these equations solved for the components of $Y$ to obtain
\begin{align} \label{eqn:Y-ab}
\begin{aligned}
Y_x =& - \frac{2 a r}{\Psi} \alpha - \frac{2 a^2 x}{\Psi} \beta, \\
Y_{\phi} =& - \frac{2 a^2 \nu x f}{ \Psi^2} \alpha + \frac{2 a \nu r f}{\Psi^2} \beta.
\end{aligned}
\end{align}

The next step will be to decompose this system using Fourier analysis in the $\phi$-variable. We will use a hat to denote the Fourier coefficient; for example,
\begin{equation} 
\alpha = \sum_{k \in \mathbb{Z}} \widehat{\alpha}_k(x) e^{i k \phi},  \; \; \; \mbox{where} \; \;\; \widehat{\alpha}_k(x) := \frac{1}{2\pi} \int_{-\pi}^{\pi} e^{-i k \phi} \alpha(x, \phi) d\phi.
\end{equation}

Under this decomposition, the Fourier modes obey
\begin{align} \label{eqn:fourier-modes}
\begin{gathered}
0=\partial_x \left( \frac{f}{\Psi} \partial_x \halpha_k \right) -  \frac{\Psi}{\nu^2 f} k^2 \halpha_k +  \frac{1}{2 \nu} ( ik \hY_{x} - \partial_{x} \hY_{\phi}),\\
\mp 12 \hv_k=\partial_x \left( \frac{f}{\Psi} \partial_x \hbeta_k \right) - \frac{\Psi}{\nu^2 f} k^2 \hbeta_k - \frac{3}{2} \left( \partial_x \left(\frac{f}{\Psi} \hY_x\right) + \frac{\Psi}{\nu^2 f} ik \hY_{\phi} \right) + \frac{4 a^2 f}{\Psi^2} \hbeta_k,\\
(\hv_k)_x = -\frac{1}{2} \frac{\Psi}{f \nu} ik \halpha_k + \frac{1}{2} (\hbeta_k)_x - \frac{1}{4} \hY_x,\\
ik \hv_k = \frac{1}{2} \frac{\nu f}{\Psi} (\halpha_k)_x + \frac{1}{2} ik \hbeta_k - \frac{1}{4} \hY_{\phi}.
\end{gathered}
\end{align}
We have elected to not indicate the mode number $k$ on the transforms of the vector field components of $Y$ to lighten the notation.  Note that we also use subscripts with $x$ to denote partial derivatives in $x$.

The Fourier transform of the equations \eqref{eqn:Y-ab} are
\begin{align} \label{eqn:hY-ab}
\begin{aligned}
\hY_x &= - \frac{2 a r}{\Psi} \halpha_k - \frac{2 a^2 x}{\Psi} \hbeta_k,\\
\hY_{\phi} &= - \frac{2 a^2 \nu x f}{ \Psi^2} \halpha_k + \frac{2 a \nu r f}{\Psi^2} \hbeta_k.
\end{aligned}
\end{align}
We insert the expressions of equation \eqref{eqn:hY-ab} into system \eqref{eqn:fourier-modes}.  After simplification we obtain
\begin{align} \label{eqn:fourier-modes2}
\begin{gathered}
0=\partial_x \left( \frac{f}{\Psi} \partial_x \halpha_k \right) -  \frac{\Psi}{\nu^2 f} k^2 \halpha_k + \partial_x\left( \frac{a^2 x f}{\Psi^2} \halpha_k - \frac{a r f}{\Psi^2} \hbeta_k \right) - ik \frac{a}{\nu \Psi} \left( r \halpha_k + a x \hbeta_k\right) ,\\
\mp 12 \hv_k=\partial_x \left( \frac{f}{\Psi} \partial_x \hbeta_k \right) - \frac{\Psi}{\nu^2 f} k^2 \hbeta_k + 3\partial_x \left(\frac{a r f}{\Psi^2} \halpha_k + \frac{a^2 x f}{\Psi^2} \hbeta_k \right) + 3\left( \frac{ik a^2 x}{\nu \Psi} \halpha_k - \frac{ik a r}{\nu \Psi} \hbeta_k \right)   + \frac{4 a^2 f}{\Psi^2} \hbeta_k,\\
(\hv_k)_x = -\frac{1}{2} \frac{\Psi}{f \nu} ik \halpha_k + \frac{1}{2} (\hbeta_k)_x + \frac{ a r}{2\Psi} \halpha_k + \frac{ a^2 x}{2\Psi} \hbeta_k,\\
ik \hv_k = \frac{1}{2} \frac{\nu f}{\Psi} (\halpha_k)_x + \frac{1}{2} ik \hbeta_k + \frac{ a^2 \nu x f}{2 \Psi^2} \halpha_k - \frac{ a \nu r f}{2\Psi^2} \hbeta_k.
\end{gathered}
\end{align}

We have the following boundary conditions for each Fourier mode.
\begin{align}
\halpha_k(\pm 1) = 0 = \hbeta_k(\pm 1), \; \; \; \; \forall k \in \mathbb{Z}.
\end{align}

\subsection{No odd-order modes} \label{sec:noodd}

In this section we prove that the combination of boundary conditions on $\halpha_k,$ and $\hbeta_k$, as well as the assumption of analyticity implies that there are no odd-order modes.

\begin{theorem}[No odd-order modes] \label{thm:no-odd-mode} If $k = 2m + 1$ for $m \in \mathbb{N}$ and $\halpha_k, \hbeta_k, \hv_k$ are analytic solutions to system \eqref{eqn:fourier-modes2} that satisfy the boundary conditions
\begin{align}
\halpha_k( \pm 1 ) = 0 = \hbeta( \pm 1 )	
\end{align}
then $\halpha_k \equiv 0, \hbeta_k \equiv 0, \hv \equiv 0$ on $[-1,1]$.
\end{theorem}

\begin{proof}
Introduce the quantity $Z(x) := \frac{f(x)}{\Psi(x)}$. Now $Z(x)$ is a rational function of $x$ that vanishes simply at $x = \pm 1$.  There are also simple poles at $x = \pm \frac{r}{a} i$, and simple zeroes at either $x = \pm \frac{1}{a}$ (Kerr-dS) and $x = \pm \frac{i}{a}$ (Kerr-AdS case).  

It will be helpful to rewrite system \eqref{eqn:fourier-modes2} one more time.  Inserting the fourth member of the system into the second, and multiplying the first two equations by $Z(x)$ yields a singular system of ordinary differential equations for $\halpha_k, \hbeta_k$ (and decoupled from $\hv_k$) which reads
\begin{eqnarray} \label{eqn:fourier-modes3}
(Z \partial_x)^2 \halpha_k  -  \frac{k^2}{\nu^2} \halpha_k &=& - Z \partial_x\left( \frac{a^2 x}{\Psi} Z \halpha_k - \frac{a r }{\Psi} Z \hbeta_k \right) + ik Z \frac{a}{\nu \Psi} \left( r \halpha_k + a x \hbeta_k\right),\\
(Z \partial_x)^2 \hbeta_k  - \frac{k^2}{\nu^2} \hbeta_k &=& -3 Z \partial_x \left(\frac{a r }{\Psi} Z \halpha_k + \frac{a^2 x }{\Psi} Z \hbeta_k \right) - 3 Z \left( \frac{ik a^2 x}{\nu \Psi} \halpha_k - \frac{ik a r}{\nu \Psi} \hbeta_k \right)  \\
&&- \frac{4 a^2}{\Psi} Z^2 \hbeta_k  \pm \left ( \frac{6 \nu}{ik} Z^2 \partial_x \halpha_k +  6 Z \hbeta_k + \frac{ 6a^2 \nu x}{ik \Psi} Z^2 \halpha_k - \frac{ 6 a \nu r }{ik\Psi} Z^2 \hbeta_k\right ). \nonumber
\end{eqnarray}
The sign difference in the last four terms arising from the dS and AdS cases will not play a role in the remainder of the proof.

Since we assume the existence of an analytic solution $\alpha_k$ and $\beta_k$ that already vanishes at $x=-1$, we may expand the solution in as a series
\begin{align}
\halpha_k(x) = \sum_{j = 1} \halpha_{k,j} (x+1)^{j}, \; \; \mbox{and} \; \; \hbeta_k(x) = \sum_{j = 1} \hbeta_{k,j} (x+1)^{j}.
\end{align}

We wish to use the system to determine a recurrence relation for the Taylor coefficients.  Begin with the left-hand side of equation \eqref{eqn:fourier-modes3}.  Observe that $Z \partial_x$ does not lower the order of vanishing in either powers of $x+1$ or $x-1$, indeed computing an expansion (in powers of $x+1$ for example) we find to leading order in both the dS and AdS cases that
\begin{align}
Z \partial_x \left[ c_{j} (x+1)^{j} \right] = \frac{2}{\nu} \, j \,  c_{j} \, (x+1)^{j} + O(\, (x+1)^{j+1} \, ).
\end{align}
Thus the operator appearing on the left-hand side of equation \eqref{eqn:fourier-modes3} acts on a Taylor coefficient by
\begin{align}
\left( (Z \partial_x)^2 - \frac{k^2}{\nu^2} \right) \left[ c_{j} (x+1)^{j} \right] = \frac{1}{\nu^2} (4 j^2 - k^2) \,  c_{j} \, (x+1)^{j} + O(\, (x+1)^{j+1} \, ).
\end{align}

Now consider the right-hand size of equation \eqref{eqn:fourier-modes3}.  Each term when applied to $c_j (x+1)^j$ vanishes to order higher than $j$ at $x=\pm 1$ simply because products of $Z$ with either $\halpha_k$ or $\hbeta_k$ vanish to one order higher than $\halpha_k$ or $\hbeta_k$, and $Z \partial_x$ does not lower the order of vanishing of a Taylor series in $x+1$ or $x-1$. Thus we may inductively conclude that the power series of both $\halpha_k$ or $\hbeta_k$ computed at either $x= \pm 1$ vanish so long as the coefficient $4 j^2 - k^2$ is never zero.  This is assured when $k$ is odd, and so we conclude the power series for both of these functions vanishes at either pole.  

The power series for either $\halpha_k$ or $\hbeta_k$ centred at $x=\pm 1$ has radius of convergence determined by the distance from $x=\pm 1$ to the nearest pole of a coefficient of the equation.  In this case, the radius of convergence is $\sqrt{1 + \frac{r^2}{a^2}} > 1$, and thus the series at $x=-1$ is analytically extended by the series at $x=1$ due to the overlap on a symmetric open interval about $x=0$.  We conclude $\halpha_k$ or $\hbeta_k$ both vanish on $[-1,1]$.  By the fourth member of system \eqref{eqn:fourier-modes2} we conclude $\hv_k$ vanishes as well.
\end{proof}

\begin{remark}
The argument above also shows that higher even modes with $k = 2m$ vanish to order $m-1$ at $x= \pm 1$. 
\end{remark}

\begin{remark}
The argument above also adapts to the extreme Kerr near-horizon geometry studied in \cite{CST18,JK13}.  See the Appendix for details.
\end{remark}

\subsection{ The $k=0$ mode } \label{sec:zeromodes}

The system \eqref{eqn:fourier-modes2} for the zero mode reads

\begin{align} \label{eqn:zero-mode-eqns}
\begin{aligned}
0=&\, \partial_x \left( \frac{f}{\Psi} \partial_x \halpha_0 \right)  + \partial_x\left( \frac{a^2 x f}{\Psi^2} \halpha_0 - \frac{a r f}{\Psi^2} \hbeta_0 \right) ,\\
\mp 12  \hv_0 =&\, \partial_x \left( \frac{f}{\Psi} \partial_x \hbeta_0 \right)  + 3\partial_x \left(\frac{a r f}{\Psi^2} \halpha_0 + \frac{a^2 x f}{\Psi^2} \hbeta_0 \right) + \frac{4 a^2 f}{\Psi^2} \hbeta_0,\\
(\hv_0)_x =&\, \frac{1}{2} (\hbeta_0)_x + \frac{ a r}{2\Psi} \halpha_0 + \frac{ a^2 x}{2\Psi} \hbeta_0,\\
0 =&\, \frac{1}{2} \frac{\nu f}{\Psi} (\halpha_0)_x + \frac{ a^2 \nu x f}{2 \Psi^2} \halpha_0 - \frac{ a \nu r f}{2\Psi^2} \hbeta_0.
\end{aligned}
\end{align}

The first equation is implied by the fourth. The first-order perturbation of the conformal variation $\hat{v}_0$ does not decouple from the variables $\hat{\alpha}_0$ and $\hat{\beta}_0$ making it less amenable to arguments used in the $k \neq 0$ case. The zero mode corresponds to axisymmetric linearized perturbations of Kerr-(A)dS$_4$ NHG. We find immediately that
\begin{equation}
    \hbeta_0 = \frac{1}{a r} \left[ a^2 x \halpha_0 +  (\halpha_0)_x \Psi \right].
\end{equation} 
Recall that from the boundary conditions, $\hbeta_0(\pm 1) =0$. It then follows from the requirement $\halpha_0(\pm 1) =0$ that we must also have $(\halpha_0)_x(\pm 1) =0$. Substituting this into the expression for $(\hv_0)_x$ gives
\begin{equation}
  (\hv_0)_x = \frac{1}{2 a r} \left(2 a^2 \halpha_0 + 4 a^2 x (\halpha_0)_x + \Psi(\halpha_0)_{xx} \right), 
\end{equation} 
which is easily integrated
\begin{equation}
    \hv_0 = \frac{a x \halpha_0}{r} + \frac{(\halpha_0)_{x} \Psi}{2 a r} + V_0
\end{equation} 
where $V_0$ is a constant. What remains is a 3rd order linear ODE for $\halpha_0$: 
\begin{align}
\begin{aligned}
0 = &  \left[ f \Psi^2 \right] (\halpha_0)_{xxx} +  \left[ (f \Psi^2)^{'} + 2 a^2 x  f \Psi \right] (\halpha_0)_{xx}  +  \left[ 2 f \left(a^4 x^2+7 a^2 r^2\right)+ 6 \Psi \left( a^2 x f' \pm \Psi^2 \right) \right] (\halpha_0)_{x} \\ & +  \left[ 4 a^2 \left( \Psi f' \pm 3  x \Psi^2 - a^2 x f\right)\right] \halpha_0  \pm 12 V_0 a r \Psi^2 
\end{aligned}
\end{align}

which can be written in the following form :
\begin{equation}
    (x+1)^3 P_3(x) y''' + (x+1)^2 P_2(x) y'' + (x+1) P_1(x) y' + P_0(x) y + Q(x)  = 0
\end{equation}
where,
\begin{equation} 
\begin{aligned}
P_3(x) :=& (1-x)(1 \pm a^2x^2) \Psi^2 , \\ 
P_2(x) :=& (f \Psi^2)^{'} + 2 a^2 x  f \Psi   \\
P_1(x) :=& (x+1) \left[ 2 f \left(a^4 x^2+7 a^2 r^2\right)+ 6 \Psi \left( a^2 x f' \pm  \Psi^2 \right) \right] \\ 
P_0(x) :=& 4a^2 (1+x)^2   \left( \Psi f' \pm 3 x \Psi^2 - a^2 x f\right) \\ Q(x) := & \pm 12 (1+x)^2 V_0 a r \Psi^2. 
\end{aligned}
\end{equation} 

This is a third-order linear inhomogeneous ODE with a regular singular point at $x=-1$. On general grounds, given the behavior of the coefficients of derivative terms, we can expect to find three linearly independent solutions on the interval $x\in (-1,1)$, and at most three such solutions on $[-1,1]$.  We are interested in those solutions which extend continuously to $x = -1$ and obey the boundary conditions. We search for Frobenius series solutions of the homogeneous differential equation, which take the form
\begin{equation}
    y(x) = \sum_{n=0} b_n (x+1)^{n+p}.
\end{equation}

Noting that $P_3(-1) = P_2(-1) = 2(1-a^2 \delta)\Psi^2$ and $P_1(-1) = P_0(-1) =0$, the lowest ($n=0$) term gives the indicial equation
\begin{equation}
    p(p-1)^2 P_3(1) x^p =0,
\end{equation} which has roots $p_1 =0$ and $p_2=1$ with multiplicity 2. There will be one series solution associated to $p_2 =1$ (the second linearly independent solution associated to this root will necessarily have a $\log|x+1|$ term which diverges at the pole). This solution will be of the form
\begin{equation}\label{sol1}
    y_1(x) = \sum_{n=0} a_n (x+1)^{n+1} = \sum_{n=1} a_n (x+1)^{n+1}.
\end{equation} where we have imposed the requirement that $y_1'(1) =0$. The solution so produced is determined by one free parameter (which we can take to be $a_1$). We now investigate the possibility of a second linearly independent solution associated to the root $p_1 =0$,
\begin{equation}\label{sol2}
    y_0(x) = \sum_{n=0} b_n (x+1)^n, 
\end{equation} satisfying the boundary conditions.  Noting that $p_2 - p_1 =1$, we are in a situation where the indices differ by an integer. Then either we produce a second solution, or we merely reproduce $y_1(x)$. We require immediately that $b_0 = b_1 =0$. But then \eqref{sol2} is identical to \eqref{sol1}. Hence there will be at most a single admissible one parameter series solution $\halpha_0(x) = y_1(x) + y_P(x)$ which converges on $x \in [-1,1]$ where $y_P$ is a particular solution of the inhomogeneous equation, which satisfies the boundary conditions.  It follows that we must identify this one-parameter family of axisymmetric solution with linearisations lying tangent to the curve of solutions within the Kerr-AdS family given by \eqref{KerrAdScurve}. In particular there are no other allowed zero modes.

\subsection{A global argument for the Kerr-(A)dS NHG} \label{sec:global}

In this section we adapt the global argument given in \cite{JK13} to the Kerr-(A)dS setting. The argument uses integration by parts and an estimation of various $L^2$ norms of terms in system \eqref{eqn:fourier-modes2}. 

For complex-valued $\halpha_k$ and $\hbeta_k$, introduce the vector quantity
\begin{equation}
w = \begin{bmatrix} \halpha_k \\ \hbeta_k \end{bmatrix} 
\end{equation}
and an $L^2$ hermitian inner product on such complex vector-valued functions on $[-1,1]$ by
\begin{equation}
(p,q) = \int_{-1}^1 \overline{p(x)}^T q(x) dx,
\end{equation}
with corresponding $L^2$ norm $\|p\| = \sqrt{(p,p)}$. Finally, introduce the operator $ \bX :=  \sqrt{\frac{f}{\Psi}} \frac{d}{dx}$ acting componentwise. Define
\begin{equation}
\bX^* (\cdot):= - \frac{d}{dx} \left( \sqrt{\frac{f}{\Psi}} \; \cdot \right), 
\end{equation}
with Dirichlet boundary conditions at $x = \pm 1$, so that for $p,q$ sufficiently regular we have $(p, \bX^* q) = (\bX p, q)$. With this notation, for $k \neq 0$ we may use the fourth equation of system \eqref{eqn:fourier-modes2} to solve for $\hv_k$, and insert the result into the second equation.  We may then rewrite the first two equations as 
\begin{align} \label{eqn:func-arg1}
-\bX^* \bX w- \frac{ k^2}{\nu^2} \frac{\Psi}{f} I w - \bX^* (M_1 w)+ i k M_2 w  + M_3 w + \frac{i}{k} M_4 \bX w - \frac{i}{k} M_5 w = 0,
\end{align}
where we have introduced
\begin{align}
\begin{aligned}
M_1 &:= \frac{a }{\Psi} \sqrt{\frac{f}{\Psi}} \begin{bmatrix} a x & - r \\ 3r & 3a x \end{bmatrix}\\
M_2 &:=   \frac{a  }{\nu \Psi} \begin{bmatrix} -r &  -ax \\ 3ax & -3r \end{bmatrix} \\
M_3 &:= \begin{bmatrix}  0 & 0 \\ 0 & \frac{4 a^2 f}{\Psi^2} \pm 6 \end{bmatrix} \\
M_4 &:= \mp 6 \nu  \sqrt{\frac{f}{\Psi}} \begin{bmatrix} 0 & 0 \\ 1 & 0 \end{bmatrix} \\
M_5 &:= \mp 6 a  \nu \frac{f}{\Psi^2} \begin{bmatrix} 0 & 0 \\ -ax & r \end{bmatrix}.
\end{aligned}
\end{align}

Note that in $M_3, M_4, M_5$ above, the sign ambiguity is resolved by choosing the upper sign for the Kerr-dS metric and the lower sign for Kerr-AdS metric.  This ambiguity does not play a role in our main result. 

We now state and prove one of our main results.  Recall that the Kerr-(A)dS metrics are specified by a choice of either $a$ or $r$ as parameter, as the other value is constrained by equations \eqref{KerrAdScurve} or \eqref{KerrdScurve}.

\begin{theorem} \label{thm:no-big-modes}  For fixed value of the parameter $a$ or $r$, there are no nontrivial solutions to system \eqref{eqn:fourier-modes2} (equivalently \eqref{eqn:func-arg1}) for $|k|$ sufficiently large.
\end{theorem}

\begin{proof}
Suppose that $w$ is a nontrivial solution to equation \eqref{eqn:func-arg1}. We will show that this places a restriction on $|k|$.  Pairing this equation with $w$ and integrating by parts one may obtain
\begin{equation}
\begin{split}
0=&\, -\left\|\bX w\right\|^2 - \frac{ k^2}{\nu^2} \left\| \sqrt{\frac{\Psi}{f}} w \right\|^2 - (\bX w, M_1w) + ik (w, M_2w) + (w, M_3 w)\\
&\, + \frac{i}{k} \overline{(M_4 w, \bX w)} - \frac{i}{k} (w, M_5 w).
\end{split}
\end{equation}
Now set $x := \| \bX w\|/\|w\|$ and $y:=  \left\| \frac{1}{\nu}\sqrt{\frac{\Psi}{f}} w \right\|/\|w\|$, and rewrite the equation as 
\begin{equation}
\begin{split}
\|w\|^2 x^2 + k^2 \|w\|^2 y^2 =&\, - ( \bX w, M_1w) + ik \left( \frac{1}{\nu} \sqrt{\frac{\Psi}{f}} w, \widetilde{M_2} w\right)\\
&\, + (w, M_3 w) +\frac{i}{k} \overline{(M_4 w,  \bX w)} - \frac{i}{k} (w, M_5 w),
\end{split}
\end{equation}
where we have introduced $\widetilde{M_2} =  \nu \sqrt{\frac{f}{\Psi}} M_2$ for convenience.  Applying the Cauchy-Schwarz inequality we obtain
\begin{align}
\begin{aligned}
&\, \|w\|^2 x^2 + k^2 \|w\|^2 y^2 \\
\leq &\,  |( \bX w, M_1w)| + |k| \left|\left( \frac{1}{\nu}  \sqrt{\frac{\Psi}{f}} w, \widetilde{M_2}  w\right)\right| + |(w, M_3 w)| +\frac{1}{|k|} |(M_4 w, \bX w)| + \frac{1}{|k|} |(w, M_5 w)| \\
\leq &\, |( \bX w, M_1w)| + |k| \left| \left( \frac{1}{\nu} \sqrt{\frac{\Psi}{f}} w,\widetilde{M_2}  w \right) \right| + |(w, M_3 w)| + |(M_4 w, \bX w)| + |( w, M_5 w)| \\
\leq &\, x \| M_1 \| \|w\|^2  + y |k|  \left\| \widetilde{M_2}  \right\| \|w\|^2 + \| M_3 \| \|w\|^2 + x \|M_4\| \|w\|^2 + \| M_5 \| \|w\|^2,
\end{aligned}
\end{align}
where we have used the fact that $|k| \geq 1$. We therefore obtain
\begin{align}
 x^2 + k^2 y^2 & \leq x (\| M_1 \| +\|M_4\| )  + y |k|  \left\| \widetilde{M_2} \right\|  + \| M_3 \|  + \| M_5 \|.
\end{align}
We can write this inequality as
\begin{align}
\left( x - \frac{ \|M_1\| + \|M_4\| }{2} \right)^2 + \left( |k| y - \frac{ \| \widetilde{M_2} \| }{2}\right)^2 \leq \frac{ (\|M_1\| + \|M_4\|)^2}{4} + \frac{ \| \widetilde{M_2} \|^2 }{4} + \|M_3\| + \|M_5\|.
\end{align}
It follows that 
\begin{align}
\left( |k| y - \frac{ \| \widetilde{M_2} \| }{2}\right)^2 \leq \frac{ (\|M_1\| + \|M_4\|)^2}{4} + \frac{ \| \widetilde{M_2} \|^2 }{4} + \|M_3\| + \|M_5\|,
\end{align}
and thus
\begin{align} \label{eqn:bound-on-k}
 |k|  \leq  \frac{\| \widetilde{M_2} \| }{2y} + \, \frac{1}{y} \sqrt{\frac{ (\|M_1\| + \|M_4\|)^2}{4} + \frac{ \| \widetilde{M_2} \|^2 }{4} + \|M_3\| + \|M_5\|},
\end{align}

We proceed to establish an upper bound on $|k|$ by bounding $y$ from below and all remaining quantities from above.  Recalling the definition of $y$,
\begin{equation}
    y =  \left\| \frac{1}{\nu}\sqrt{\frac{\Psi}{f}} w \right\|/\|w\| \geq \frac{1}{\nu} \inf_{x \in [-1,1]} \sqrt{\frac{\Psi}{f}}.
\end{equation}
In the Kerr-AdS case, $y \geq \frac{r}{\nu}$, while in the Kerr-dS case, $y \geq \frac{r}{ \nu \sqrt{1+a^2} }$, which in either case is a positive constant depending on the parameter.

As for the matrix norms, recall the operator norm on matrices is induced by the inner product on $L^2([-1,1];\mathbb{C}^2) $ as is given by
\begin{align}
\| M \| := \sup\left\{  \| M z\| : \, z \in L^2([-1,1];\mathbb{C}^2), \, \|z\|=1 \right\}.
\end{align}
For each of the matrices $M_1, \widetilde{M_2}, M_3, M_4, M_5$ defined above the components are continuous functions of $x$.  Thus each norm may be bounded a constant depending on the parameter.  Equation \eqref{eqn:bound-on-k} provides a bound on the number of nontrivial solutions to equation \eqref{eqn:fourier-modes2}, completing the proof.
\end{proof}

In view of Theorem \ref{thm:no-big-modes}, it is natural to estimate the maximum number of nontrivial Fourier modes explicitly for certain values of the parameter.  Here we split the Kerr-dS and Kerr-AdS cases as a favorable sign appears to strengthen the result for the Kerr-dS case.

We will explicitly estimate the matrix norms as follows.  If $M := \begin{bmatrix} m_1 & m_2 \\ m_3 & m_4 \end{bmatrix}$ is a matrix of real-valued coefficient functions and 
$z := \begin{bmatrix}
     z_1 \\ z_2
 \end{bmatrix}$ is complex-valued, then
\begin{align}
\begin{aligned}
     \| M \|^2 & = \sup\left\{  \| M z\| : \, z \in L^2([-1,1];\mathbb{C}^2), \, \|z\|=1 \right\}
 \\ & = \sup\left\{  \int_{-1}^{1} \left| m_1 z_1 + m_2 z_2 \right|^2 + \left| m_3 z_1 + m_4 z_2 \right| ^2 \, dx:  \, z \in L^2([-1,1];\mathbb{C}^2), \, \| z \| = 1  \right\} \\ & \leq  \max \left\{ \sup_{x \in [-1,1]} m_1^2 + m_3^2 + |m_1 m_2 + m_3 m_4 |,  \sup_{x \in [-1,1]} m_2^2 + m_4^2 + |m_1 m_2 + m_3 m_4 | \right\}.
\end{aligned}
\end{align}

\subsubsection{The Kerr-dS NHG}  Choose $r$ as the  parameter.  Recall that $ 0 < r < \frac{1}{\sqrt{3}}$, and $0 < a^2 < r^2 < \frac{1}{3}$.  One easily checks that for all $x \in [-1,1]$,
\begin{align}
0 \leq f(x) \leq 1, \quad  \frac{f(x)}{\Psi(x)} \leq \frac{1}{r^2}, \quad \, \text{ and } \, r^2 \leq  \Psi(x) \leq r^2 + a^2.
\end{align}
Then for example for $M_1$,
\begin{align}
\begin{aligned}
m_1^2 + m_3^2 + |m_1 m_2 + m_3 m_4 | &= \left( \frac{a}{\Psi} \sqrt{\frac{f}{\Psi}} \right)^2 [a^2 x^2 + 9 r^2 + 8 a r |x|], \; \mbox{and},\\
m_2^2 + m_4^2 + |m_1 m_2 + m_3 m_4 | &= \left( \frac{a}{\Psi} \sqrt{\frac{f}{\Psi}} \right)^2 [r^2 + 9 a^2 x^2  + 8 a r |x|].
\end{aligned}
\end{align}
Applying the estimates above to bound each of these functions of $x$ on $[-1,1]$ and then taking the maximum estimate we find
\begin{align}
\| M_1 \| \leq \frac{3 \sqrt{2}}{r}.
\end{align}

For $\widetilde{M_2}$,  we find
\begin{align}
\begin{aligned}
m_1^2 + m_3^2 + |m_1 m_2 + m_3 m_4 | &= \left( \frac{a}{ \Psi}  \sqrt{\frac{f}{\Psi}}\right)^2 [r^2 + 9 a^2 x^2 + 8 a r |x|], \; \mbox{and},\\
m_2^2 + m_4^2 + |m_1 m_2 + m_3 m_4 | &= \left( \frac{a}{ \Psi}  \sqrt{\frac{f}{\Psi}}\right)^2 [a^2 x^2 + 9 r^2  + 8 a r |x|].
\end{aligned}
\end{align}
This results in extremization of the same quantities as $M_1$.  So
\[ \| \widetilde{M_2} \| \leq \frac{3 \sqrt{2}}{r}. \]

For $M_3$, 
\begin{align}
\begin{aligned}
m_1^2 + m_3^2 + |m_1 m_2 + m_3 m_4 | &= 0 \; \mbox{and},\\
m_2^2 + m_4^2 + |m_1 m_2 + m_3 m_4 | &= \left(\frac{4 a^2 f}{\Psi^2} + 6\right)^2.
\end{aligned}
\end{align}
Thus we obtain
\begin{align}
\| M_3 \| \leq \frac{4}{r^2} + 6 \leq \frac{6}{r^2}.
\end{align}

For $M_4$,
\begin{align}
\begin{aligned}
m_1^2 + m_3^2 + |m_1 m_2 + m_3 m_4 | &= 36 \nu^2 \frac{f}{\Psi} \; \mbox{and},\\
m_2^2 + m_4^2 + |m_1 m_2 + m_3 m_4 | &= 0.
\end{aligned}
\end{align}
We obtain,
\begin{align}
\| M_4 \| \leq \frac{6 \nu}{r}.
\end{align}

Finally, for $M_5$,
\begin{align}
\begin{aligned}
m_1^2 + m_3^2 + |m_1 m_2 + m_3 m_4 | &= \left( 6 a \nu \frac{f}{ \Psi^2}\right)^2 [ a^2 x^2 + ar |x|], \; \mbox{and},\\
m_2^2 + m_4^2 + |m_1 m_2 + m_3 m_4 | &=\left( 6 a \nu \frac{f}{ \Psi^2}\right)^2 [ r^2 + ar |x|],
\end{aligned}
\end{align}
and we obtain 
\begin{align}
\| M_5 \| \leq \frac{6 \sqrt{2} \nu}{r^2} .
\end{align}

Finally, recall from the proof of Theorem \ref{thm:no-big-modes} that the nonzero mode numbers are bounded above by the inequality of \eqref{eqn:bound-on-k}.  Thus, noting that the quantity $y \geq r/\nu$, we find
\begin{align} 
|k|  &\leq  \frac{1}{y} \left( \frac{\| \widetilde{M_2} \| }{2} + \, \sqrt{\frac{ (\|M_1\| + \|M_4\|)^2}{4} + \frac{ \| \widetilde{M_2} \|^2 }{4} + \|M_3\| + \|M_5\|} \right)\\
& \leq \frac{\nu}{r} \left( \frac{3 \sqrt{2}}{2r} + \, \sqrt{\frac{ ( 3 \sqrt{2}  + 6 \nu )^2}{4r^2} + \frac{ 18  }{4r^2} + \frac{6}{r^2} + \frac{6\sqrt{2} \nu}{r^2}} \right)\\
& \leq \frac{\nu}{r^2} \left( \frac{3 \sqrt{2}}{2} + \, \sqrt{\frac{ ( 3 \sqrt{2}  + 6 \nu)^2}{4} + \frac{ 18  }{4} + 6 + 6\sqrt{2}} \nu \right).
\end{align}
One may estimate the maximum value of $|k|$ using $\nu/r^2 \leq 2$ and $\nu \leq 2/3$.  One may also insert use the relation between $r$ and $a$ in equation \eqref{KerrdScurve}, and then a computer algebra packagage to find the maximum value over $r \in (0, 1/\sqrt{3})$.  Using Mathematica we found that $|k| \leq 11.99$.  Since $k$ must be an integer, we conclude that a putative nontrivial mode must satisfy $|k| \leq 11.$  

\subsubsection{The Kerr-AdS NHG}  In this case recall that $ 0 < r^2 < a^2 < 1$.  Again on $x \in [-1,1]$ we have
\begin{align}
0 \leq f(x) \leq 1, \quad  \frac{f(x)}{\Psi(x)} \leq \frac{1}{r^2}, \quad \, \text{ and } \, r^2 \leq  \Psi(x) \leq r^2 + a^2.
\end{align}
The estimation of the matrix norms proceeds much as before, and in fact we obtain the same bounds for the matrix norm of each of $M_1, \widetilde{M_2}, M_4$ and $M_5$.  For $M_3$, we adjust the estimate slightly to obtain
\begin{align}
\| M_3 \| \leq \left|\frac{4 a^2 f}{\Psi^2} - 6\right| \leq \frac{4 a^2 f}{\Psi^2} + 6  \leq \frac{12}{r^2} + 6 \leq \frac{14}{r^2},
\end{align}
where we used that $ r^2 < 1/3$.

Estimating the maximal $|k|$ as before we obtain
\begin{align} 
|k|  &\leq  \frac{1}{y} \left( \frac{\| \widetilde{M_2} \| }{2} + \, \sqrt{\frac{ (\|M_1\| + \|M_4\|)^2}{4} + \frac{ \| \widetilde{M_2} \|^2 }{4} + \|M_3\| + \|M_5\|} \right)\\
& \leq \frac{\nu}{r} \left( \frac{3 \sqrt{2}}{2r} + \, \sqrt{\frac{ ( 3 \sqrt{2}  + 6 \nu )^2}{4r^2} + \frac{ 18  }{4r^2} + \frac{14}{r^2} + \frac{6\sqrt{2} \nu}{r^2}} \right)\\
& \leq \frac{\nu}{r^2} \left( \frac{3 \sqrt{2}}{2} + \, \sqrt{\frac{ ( 3 \sqrt{2}  + 6 \nu)^2}{4} + \frac{ 18  }{4} + 14 + 6\sqrt{2}} \nu \right).
\end{align}

This time we cannot expect an unrestricted bound on $k$ as $\nu \to \infty$ as $a \to 1^-$.  Using a computer algebra package, by maximizing in $r$ we find that if $a < 0.01$, then $|k| < 13.84$.  Since $k$ must be an integer, we conclude that a putative nontrivial mode must satisfy $|k| \leq 13.$

\subsection{Proof of Theorem \ref{theorem3}} We now prove the main Theorem.

Any analytic one-parameter family $(M, g(s), X(s))$ of NHGs with $s=0$ corresponding to a Kerr-(A)dS metric can be expressed using the Jezierski-Kami\'nksi formalism of \S \ref{sec:JK} as a nontrivial solution to the coupled system \eqref{eq5.2} with boundary conditions $\alpha(\pm 1) = 0 = \beta(\pm 1)$.  Computing a Fourier transform in the axial variable, we find that Fourier modes satisfy the equations of system \eqref{eqn:fourier-modes2}.  Theorem \ref{thm:no-big-modes} then implies that there are no nontrivial solutions to this equation for $|k|$ sufficiently large.  This proves that the space of solutions of the linearization of the NHG equation is finite-dimensional.

To estimate the dimension of the space of deformations in the Kerr-dS case, recall there are no modes for $k$ odd, and further solutions of equation \eqref{eq5.2} are real-valued so that $a_{-k} = \overline{a_k}$. Calculations of the previous section restrict the nonzero modes of Kerr-dS to have  $|k|\leq 11$. Thus we conclude the space of solutions is parametrized by at most the modes with $k = 0, 2, 4, 6, 8, 10$.  The nonzero modes are complex-valued and so two dimensional over $\mathbb{R}$. We conclude the space of the solutions is at most 11 dimensional.

This concludes the proof of Theorem \ref{theorem3}.

\appendix

\section{Two remarks on the paper \cite{JK13} on  the extreme Kerr NHG}

In this appendix, we make a few observations about the extreme Kerr near horizon geometry as considered in \cite{CST18, JK13}, which was the inspiration for this paper.  We show that a calculation similar to the one given in this paper implies that the odd modes also vanish for perturbations of the extreme Kerr NHG.  We also believe a norm computation given in \cite{JK13} is incorrect.

In this appendix we redefine several quantities and defer to notation in \cite{JK13}.  Similar to our equation \eqref{eqn:fourier-modes2} above, equation (57) from \cite{JK13} represents the coupled system of ordinary differential equations satisfied by the Fourier modes $v_k$ of a perturbation of the extreme Kerr NHG:
\begin{equation} \label{app:eqn0} \partial_x \left( a^2 \partial_x v_k \right) - \frac{k^2}{a^2} v_k + D v_k + \partial_x (a^2 B v_k) + i k C v_k =0, \end{equation}
where for the extreme Kerr near horizon geometry these quantities are
\begin{itemize}
\item $\displaystyle a^2 = a(x)^2 := 2 \frac{1-x^2}{1+x^2}, \; \mbox{for} \; x \in [-1,1]$,
\item $\displaystyle v_k = \begin{bmatrix} \alpha_k \\ \beta_k \end{bmatrix}: [-1,1] \to \mathbb{C}^2$ are the Fourier modes, 
\item $\displaystyle B = \frac{1}{1+x^2} \begin{bmatrix} x & -1 \\ 3 & 3x \end{bmatrix}$,
\item $\displaystyle C = \frac{1}{1+x^2} \begin{bmatrix} 1 & x \\ -3x & 3 \end{bmatrix}$,
\item $\displaystyle D = \frac{4 a^2}{1+x^2} \begin{bmatrix} 0 & 0 \\ 0 & 1 \end{bmatrix}$.
\end{itemize}

\begin{prop}
If $v_k$ is an analytic solution to equation \eqref{app:eqn0} satisfying $v_{k}(\pm 1) = 0$, then $v_{2k-1} \equiv 0$, for all $k\in \mathbb{N}$.
\end{prop}
\begin{proof}  

Return to equation \eqref{app:eqn0} for general $k$.  We drop the $k$ index in $\alpha$ and $\beta$. Multiply the equations by $a^2$ and separate into components, to obtain
\begin{equation} \begin{cases} 
a^2 \partial_x \left( a^2 \partial_x \alpha \right) - k^2 \alpha + \phantom{\frac{4 a^4}{1+x^2} \beta} +  a^2 \partial_x \left(\frac{a^2}{1+x^2} \left(x \alpha -\beta \right) \right) + a^2 \left(\frac{ik}{1+x^2}\right) \left( \alpha + x \beta \right) &=0, \\
a^2 \partial_x \left( a^2 \partial_x \beta \right) -k^2 \beta + \frac{4 a^4}{1+x^2} \beta   + a^2 \partial_x \left(\frac{a^2}{1+x^2} \left(3 \alpha + 3x \beta \right) \right) + a^2 \left(\frac{ik}{1+x^2}\right) \left( - 3 x \alpha + 3 \beta \right) &=0. 
\end{cases}
\end{equation}

We rewrite this system as
\begin{equation} \label{eqn:main3} \begin{cases} 
(a^2 \partial_x)^2 \alpha  - k^2 \alpha &= -  a^2 \partial_x \left(\frac{a^2}{1+x^2} \left(x \alpha -\beta \right) \right) - a^2 \left(\frac{ik}{1+x^2}\right) \left( \alpha + x \beta \right), \\
(a^2 \partial_x)^2 \beta  -k^2 \beta &= - \frac{4 a^4}{1+x^2} \beta   - a^2 \partial_x \left(\frac{a^2}{1+x^2} \left(3 \alpha + 3x \beta \right) \right) - a^2 \left(\frac{ik}{1+x^2}\right) \left( - 3 x \alpha + 3 \beta \right). 
\end{cases}
\end{equation}

Recall that $a^2$ is analytic near $x=\pm 1$ and vanishes simply at these points.  Thus the operator $a^2 \partial_x$ and its iterates do not lower the order of vanishing of power series in powers of either $x+1$ or $x-1$.  On the other hand, the right-hand side has higher-order vanishing in these powers.

For simplicity we work at $x=-1$, the argument for $x=1$ is similar.  We have assumed that $\alpha$ and $\beta$ have a smooth power series solution in powers of $x+1$. We know that $\alpha(-1) = 0= \beta(-1)$, thus the power series expansion of these functions vanishes to first order.  Now assume that for some integer $n>0$ that $\alpha$ and $\beta$ vanish to order $n-1$ in $x+1$, thus
\[ \alpha = \sum_{j=n}^{\infty} \alpha_j (x+1)^j, \; \beta = \sum_{j=n}^{\infty} \beta_j  (x+1)^j. \]
Working at leading order, one finds the coefficients of the left hand side of the system \eqref{eqn:main3} to be
$(4 n^2 - k^2) \alpha_n$, and $(4n^2 - k^2) \beta_n$, respectively.  Noting that products of terms like $a^2$ and either $\alpha$ or $\beta$ vanish to order $n+1$, we see the entire right hand side vanishes to order $n+1$ or higher.  Thus if $k=0$ or $k$ is odd, then $4n^2 - k^2$ is never zero, and $\alpha_n = 0 = \beta_n$ for all $n$ by induction.  When $k$ is even, we only obtain vanishing up to a certain order depending on the mode number $k$.  The remainder of the proof is similar to the proof of Theorem \ref{thm:no-odd-mode}.
\end{proof}

Finally, the global argument given in \cite{JK13} for the elimination of Fourier modes requires the estimation of matrix norms.  We were unable to precisely replicate the computation of matrix norms in Appendix C of \cite{JK13}.  Recalling the definition of $a(x)$ and the matrix $B$ in this appendix, consider the (constant!) vector-valued function 
\[ v = \begin{pmatrix} \frac{1}{\sqrt{2}} \\ 0 \end{pmatrix}. \]
It is easy to check that $(v|v) = 1$ for the inner product in \cite[Theorem 5]{JK13}. We also have
\[ a B v = \frac{a(x)}{1+x^2}  \begin{pmatrix} \frac{x}{\sqrt{2}} \\ \frac{3}{\sqrt{2}} \end{pmatrix}. \]
A Mathematica computation yields 
\[ \| aB v \|^2 = (a Bv | a B v) = \int_{-1}^1 \frac{a(x)^2}{(1+x^2)^2} \left( \frac{x^2}{2} + \frac{9}{2} \right) dx = 5 + \pi. \]
In particular $\|a B v\| \approx 2.853$.  Recalling that an operator norm is defined by
\[ \| a B \| = \sup \{ \| aB v\|: \|v\|=1\}, \]
this calculation shows that $2.853 \leq \| a B \|.$  But the result given in \cite[p 1002]{JK13} is $\| a B \| = \sqrt{6} \approx 2.449$.

\end{document}